\documentclass[journal]{IEEEtran}
\usepackage{changes}

\usepackage{cite}
\usepackage{color}
\usepackage{graphicx}
\graphicspath{{Figs/}{../pdf/}{../jpeg/}}
\DeclareGraphicsExtensions{.pdf,.jpeg,.png}
\usepackage{lipsum}
\usepackage{amsmath}
\usepackage{amssymb}
\usepackage{mathrsfs}
\usepackage{algorithm}
\usepackage{array}
\usepackage{fixltx2e}
\usepackage{stfloats}
\usepackage{url}
\usepackage{booktabs}
\usepackage{enumerate}
\usepackage{setspace}

\usepackage{amsmath}
\usepackage{amsthm}
\newtheorem{theorem}{Theorem}[section]
\newtheorem{lemma}[theorem]{Lemma}

\newtheorem{proposition}[theorem]{Proposition}

\newtheorem{remark}[]{Remark}

\definecolor{CBLUE}{RGB}{0,112,192}
\definecolor{CRED}{RGB}{192,0,0}
\definecolor{CYELLOW}{RGB}{247,162,91}
\definecolor{CPURPLE}{RGB}{126,47,142}
\definecolor{CGREY}{RGB}{200,200,200}

\begin{document}

\title{Placing Grid-Forming Converters to Enhance Small Signal Stability of PLL-Integrated Power Systems}

\author{Chaoran Yang, Linbin Huang, Huanhai Xin, and Ping Ju
\thanks{Manuscript received July 8, 2020. The authors are with the College of Electrical Engineering at Zhejiang University, Hangzhou, China. (Emails: \text{yang\_chaoran@zju.edu.cn}, \text{huanglb@zju.edu.cn}, \text{xinhh@zju.edu.cn}, \text{pju@zju.eu.cn})}
\thanks{This work was supported by the National Natural Science Foundation of China (No. 51922094, No. 52007163).}}

%\onecolumn

\maketitle

%\doublespacing

\begin{abstract}

The modern power grid features the high penetration of power converters, which widely employ a phase-locked loop (PLL) for grid synchronization. However, it has been pointed out that PLL can give rise to small-signal instabilities under weak grid conditions. This problem can be potentially resolved by operating the converters in grid-forming mode, namely, without using a PLL. Nonetheless, it has not been theoretically revealed how the placement of grid-forming converters enhances the small-signal stability of power systems integrated with large-scale PLL-based converters. This paper aims at filling this gap. Based on matrix perturbation theory, we explicitly demonstrate that the placement of grid-forming converters is equivalent to increasing the power grid strength and thus improving the small-signal stability of PLL-based converters. Furthermore, we investigate the optimal locations to place grid-forming converters by increasing the smallest eigenvalue of the weighted and Kron-reduced Laplacian matrix of the power network. The analysis in this paper is validated through high-fidelity simulation studies on a modified two-area test system and a modified 39-bus test system. This paper potentially lays the foundation for understanding the interaction between PLL-based (i.e., grid-following) converters and grid-forming converters, and coordinating their placements in future converter-dominated power systems.

\end{abstract}


\section{Introduction}

Power converters are extensively integrated into modern power systems as the grid interfaces of renewables, HVDC systems, energy storage systems, and so on \cite{milano2018foundations,rocabert2012control,blaabjerg2006overview}. Currently, most of the converters adopt a phase-locked loop (PLL) for grid synchronization. The mechanism of PLL is easy to understand, as it can be considered as tracking the angle and frequency of the power grid with second-order dynamics \cite{golestan2015conventional}. For this reason, PLL-based converters are also widely known as ``grid-following'' converters \cite{poolla2019placement,pattabiraman2018comparison}. The large-scale integration of PLL-based converters gives rise to unprecedented changes to power systems as the synchronization mechanism of PLL-based converters is totally different from conventional power sources, i.e., synchronous generators (SGs).

Compared with conventional power systems, the synchronization dynamics of such {\em PLL-integrated power systems} are more complex and new types of instability issues may arise. For example, it has been reported that PLL-based converters could become unstable under weak grid conditions, which belongs to small-signal instability issues \cite{huang2019grid,wen2016analysis,wang2018unified,fan2018modeling}. This type of instability is dominated by the dynamics of PLL, while the oscillation frequency and stability margin are also pertinent to the design of other loops \cite{huang2019grid}. Such instability/oscillation should be prevented in practice because it endangers the secure operation of power systems and may cause economic loss once the converters get tripped.

Recent works have shown that grid-forming converters are naturally immune from the PLL-induced instabilities since a PLL is not needed to realize grid synchronization \cite{huang2017virtual,d2015virtual}. Grid-forming converters are supposed to have the capability of forming a local grid without connected to an extra voltage source. According to this definition, currently there exist many control strategies that can be classified as grid-forming, e.g., droop control, virtual synchronous machines (VSM), synchronverters, and virtual oscillator control \cite{zhong2011synchronverters,gross2019effect,d2013virtual}. Without loss of generality, in this paper we consider VSM as a prototypical type of grid-forming control, as it also covers another popular type, i.e., droop control (which can be considered as VSM with zero virtual inertia).

The initial motivations of using grid-forming control were to realize islanded operation, inertia emulation, voltage support, etc., while it turns out that another significant advantage is the robustness against various power grid strength, namely, it fits well with weak grid conditions \cite{huang2019adaptive,rocabert2012control,arghir2018grid,matevosyan2019grid}. Hence, grid-forming control can be considered as a promising technique to accommodate large-scale power converters. Moreover, grid-forming converters also achieve better performance than grid-following converters in terms of virtual inertia provision \cite{poolla2019placement}.

However, currently almost all the installed power converters in practice have been equipped with PLL-based controllers, and thus it could be unrealistic to change all of them into grid-forming converters. As an alternative, one can change some of the installed converters into grid-forming type, or require that the converters to be installed should employ grid-forming control. That is to say, future power systems will comprise both PLL-based converters and grid-forming converters.

There have been research works using a single-converter-infinite-bus system to demonstrate that grid-forming converters can maintain desired stability margin even under very weak grid conditions \cite{zhang2009power,huang2019adaptive}. By comparison, PLL-based converters may become unstable which feature sustained oscillations of the frequency output (i.e., PLL-induced instabilities) \cite{huang2019grid,wen2016analysis}. It has also attracted an increasing research interest to study the interaction between grid-forming converters and PLL-based converters \cite{milano2018foundations,poolla2019placement}. For example, Ref.~\cite{poolla2019placement} investigated how the virtual inertia implementations in PLL-based converters and grid-forming converters affect the performance of the system's frequency responses, and an optimal placement algorithm was proposed to improve the frequency responses.

However, in terms of the (PLL-induced) small-signal stability problem, it has not been studied yet how grid-forming converters and PLL-based converters interact with each other and affect such stability of multi-converter systems. In other words, it remains unclear whether the placements of grid-forming converters can help improve the small-signal stability of power systems that integrated with large-scale PLL-based converters and reduce the chance of PLL-induced oscillations. Moreover, it is also unclear how to optimally place the grid-forming converters with regards to the stability margin.

This paper aims at filling these gaps by theoretically exploring the interactions between grid-forming converters and PLL-based converters and analyzing the resulting (PLL-induced) small-signal stability.

In a first step, we model the system which describes how the grid-forming converters interact with PLL-based converters via the power network. Then, by using matrix perturbation theory, we explicitly analyze how the integration of grid-forming converters affects the small-signal stability of PLL-based converters. Our analysis is based on our previous finding that the small-signal stability of PLL-based converters is determined by the power grid strength which can be characterized by the generalized short-circuit ratio (gSCR) \cite{dong2018small,huang2019impacts}. This enables us to study the impacts of grid-forming converters on the stability of PLL-based converters by simply focusing on how the grid-forming converters equivalently change the grid strength. We will explicitly show that the integration of grid-forming converters is equivalent to placing voltage sources in the power network and thus enhance the grid strength, which is attributed to the voltage regulation inside grid-forming controls. Moreover, based on the analysis, we investigate how to optimally place the grid-forming converters to enhance the overall system stability, which is based on increasing the smallest eigenvalue of the weighted and Kron-reduced Laplacian matrix of the power network (i.e., gSCR of the system).

The three substantial contributions of this paper are summarized as follows:

\begin{enumerate}
  %\item The interaction between PLL-based converters and grid-forming converters via the network is explicitly formulated in a closed-loop dynamical system
  %\item By using matrix perturbation theory, it is theoretically shown that the placement of grid-forming converters is equivalent to enhancing the grid strength and thus has positive effects on the stability of PLL-based converters.
  \item By using matrix perturbation theory, it is theoretically shown that the placement of grid-forming converters is equivalent to changing the grid strength (characterized by gSCR) from the perspective of the PLL-based converters.
  \item It is rigorously proven that the placement of grid-forming converters has positive effects on the (PLL-induced) small-signal stability of multi-converter systems, and the locations of the grid-forming converters determine to what extent the stability can be improved.
  \item Based on the theoretical analysis of how grid-forming converters affect PLL-based converters, we formulate a tractable optimization problem in order to find the optimal locations to place grid-forming converters and enhance the (PLL-induced) small-signal stability of the system.
\end{enumerate}

The rest of this paper is organized as follows: Section II presents the system modeling considering grid-forming and PLL-based converters. Section III analyzes the impacts of grid-forming control on the stability of PLL-based converters by focusing on the grid strength. Section IV investigates the optimal placement of grid-forming converters. Simulation results are given in Section V. Section VI concludes the paper.

\section{Multi-Converter System Modeling}

In this section, we will briefly introduce the admittance models of grid-forming converters and PLL-based converters, and then develop the closed-loop model of a multi-converter system in order to illustrate how grid-forming converters interact with PLL-based converters via the network. As mentioned before, we consider VSM as a prototypical grid-forming control without loss of generality.

Fig.~\ref{Fig_converter_control} shows a three-phase converter which is connected to the ac grid via an {\em LCL} filter. The converter can be operated in PLL-based mode or grid-forming mode, with the control diagrams given in Fig.~\ref{Fig_converter_control}.

\begin{figure}[!t]
	\centering
	\includegraphics[width=3.4in]{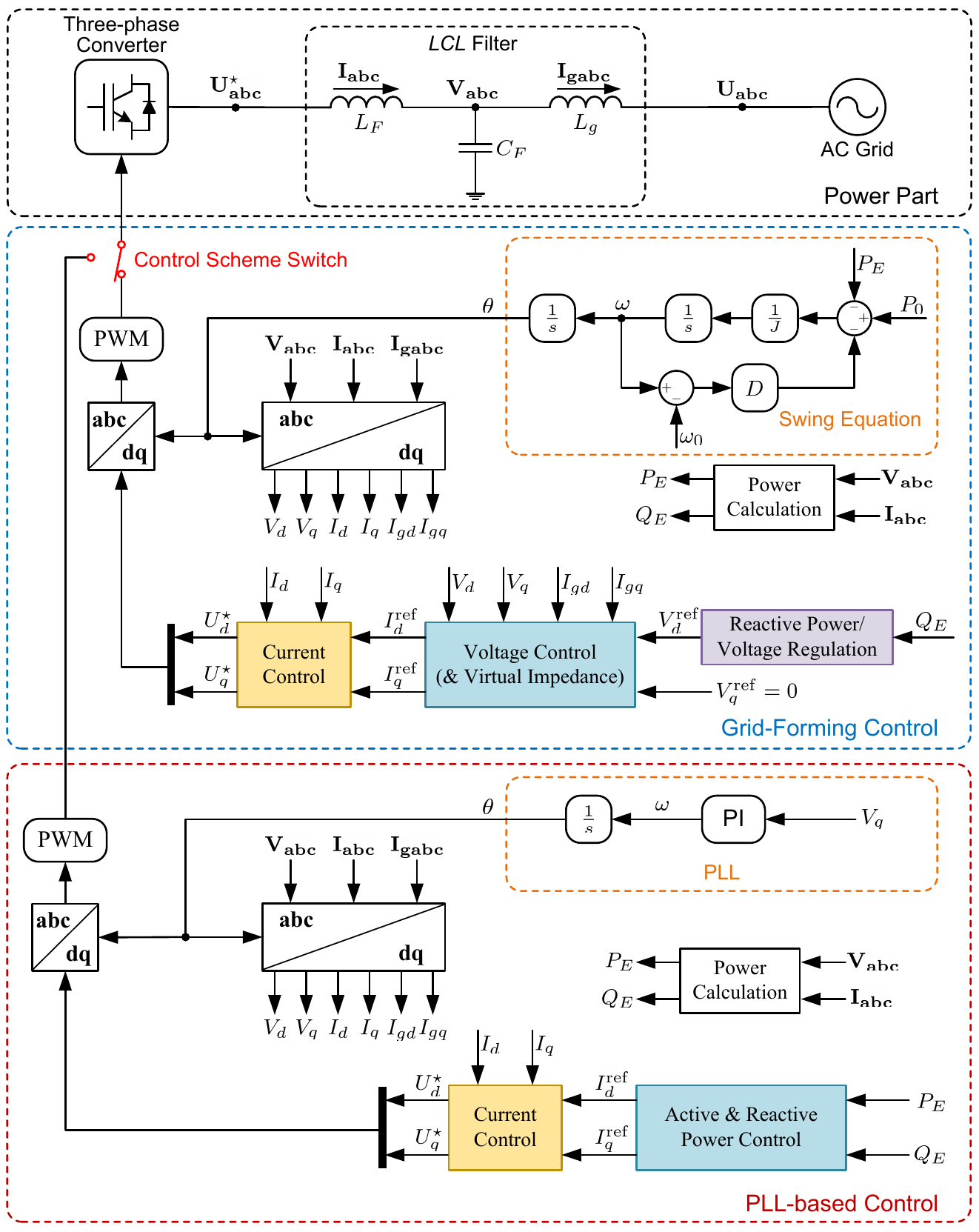}
	\vspace{-2mm}
	%\DeclareGraphicsExtensions.
	\caption{One-line diagram of a grid-connected converter with grid-forming control or PLL-based control.}
	\vspace{-2mm}
	\label{Fig_converter_control}
\end{figure}

\subsection{Admittance Modeling of PLL-based Converters}

As labelled in Fig.~\ref{Fig_converter_control}, $\bf V_{abc}$ is the three-phase capacitor voltage, $\bf I_{abc}$ is the converter-side three-phase current, $\bf I_{gabc}$ is the three-phase current injected into the ac grid, $\bf U_{abc}^\star$ is the converter's voltage determined by the modulation, and $\bf U_{abc}$ is the ac grid voltage. Let $\vec V = V_d+jV_q$, $\vec I = I_d+jI_q$, $\vec I_g = I_{gd}+jI_{gq}$, $\vec U^\star = U^\star_d+jU^\star_q$, and $\vec U = U_d+jU_q$ be respectively the space vectors of $\bf V_{abc}$, $\bf I_{abc}$, $\bf I_{gabc}$, $\bf U_{abc}^\star$, and $\bf U_{abc}$ in the {\em controller's} {\em dq}-frame.

By using complex transfer functions, the dynamic equations of the {\em LCL} can be formulated as \cite{harnefors2007modeling}
\begin{equation}\label{eq:Lf}
{\vec U^\star} - \vec V = \left( {s{L_F} + j\omega {L_F}} \right)  \vec I\,,		
\end{equation}
\begin{equation}\label{eq:Cf}
\vec I - {\vec I_g} = \left( {s{C_F} + j\omega {C_F}} \right)  \vec V  \,, 	
\end{equation}
\begin{equation}\label{eq:Lg}
\vec V - \vec U = \left( {s{L_g} + j\omega {L_g}} \right)  {\vec I_g} \,,
\end{equation}
where $L_F$ is the converter-side inductance, $L_g$ is the grid-side inductance, and $C_F$ is the capacitance of the {\em LCL} filter.

The dynamic equation of the current loop is
\begin{equation}\label{eq:current_control}
{\vec U^\star} = {\rm PI}_{\rm CC}(s)  \left( {{{\vec I}^{\rm ref}} - \vec I} \right ) + j\omega {L_F}\vec I + {f_{\rm VF}}(s)\vec V\,,	
\end{equation}
where ${\rm PI}_{\rm CC}(s)$ is the transfer function of the PI regulator, ${\vec I}^{\rm ref} = {\vec I}^{\rm ref}_d+j{\vec I}^{\rm ref}_q$ is the current reference vector which comes from the power control loops, $f_{\rm VF}(s) = K_{\rm VF}/(T_{\rm VF}s+1)$ is a first-order filter which eliminates the high-frequency components of the voltage feed-forward signals.

The power control loops can be formulated as
\begin{equation}\label{eq:power_control}
\begin{split}
I_d^{\rm ref} &= {\rm PI}_{\rm PC}(s)  \left( {{P^{\rm ref}} - {P_E}} \right)\,,\\
I_q^{\rm ref} &= {\rm PI}_{\rm QC}(s)  \left( {{Q_E} - {Q^{\rm ref}}} \right)\,,		
\end{split}
\end{equation}
where ${\rm PI}_{\rm PC}(s)$ and ${\rm PI}_{\rm QC}(s)$ are the transfer functions of the PI regulators, ${P^{\rm ref}}$ and ${Q^{\rm ref}}$ are the power reference values, $P_E$ and $Q_E$ are the active and reactive powers calculated by
\begin{equation}\label{eq:PQ}
\begin{split}
{P_E} &= {V_d}{I_{gd}} + {V_q}{I_{gq}}\,,\\
{Q_E} &= {V_q}{I_{gd}} - {V_d}{I_{gq}}.		
\end{split}
\end{equation}

The converter is synchronized with the grid via the PLL, which determines the angle of the controller's {\em dq}-frame as
\begin{equation}\label{eq:PLL}
\theta  = \frac{\omega }{s} = \frac{1}{s}  {\rm PI}_{\rm PLL}(s)  {V_q} \,,		
\end{equation}
where ${\rm PI}_{\rm PLL}(s)$ is the transfer function of the PI regulator, $\theta$ and $\omega$ are respectively the angle and angular frequency of the controller's {\em dq}-frame.

The voltage and current vectors can be transformed into the global {\em dq}-frame (whose angular frequency is $\omega_g = \omega_0 = 100\pi {\rm rad/s}$ and the angle is $\theta_g$) as
\begin{equation}\label{eq:global_dq}
\begin{split}
& \vec I'_g = I'_{gd}+jI'_{gq} = \vec I_g e^{j\delta} \,, \\
& \vec U' = U'_{d}+jU'_{q} = \vec U e^{j\delta} \,,
\end{split}
\end{equation}
where $\delta = \theta - \theta_g$, $\vec I'_g$ and $\vec U'$ are respectively the grid-side current and the grid voltage in the global {\em dq}-frame.

We note that the above equations based on space vectors and complex transfer functions can be transformed into their matrix forms by considering \cite{harnefors2007modeling}
\begin{equation}\label{eq:transf}
\begin{split}
{y_d} + j{y_q} &= \left[ {{G_d}(s) + j{G_q}(s)} \right]  \left( {{x_d} + j{x_q}} \right)\\
\Leftrightarrow \left[ {\begin{array}{*{20}{c}}
	{{y_d}}\\
	{{y_q}}
	\end{array}} \right] &= \left[ {\begin{array}{*{20}{c}}
	{{G_d}(s)}&{ - {G_q}(s)}\\
	{{G_q}(s)}&{{G_d}(s)}
	\end{array}} \right]\left[ {\begin{array}{*{20}{c}}
	{{x_d}}\\
	{{x_q}}
	\end{array}} \right].	
\end{split}
\end{equation}
Then, by linearizing \eqref{eq:Lf}-\eqref{eq:global_dq} and combining them, we obtain the admittance model of PLL-based converters denoted by
\begin{equation}\label{eq:Impedance_PLL}
- \left[ {\begin{array}{*{20}{c}}	{\Delta I'_{gd}}\\	{\Delta I'_{gq}} \end{array}} \right] = {\bf Y_{PLL}}(s) \left[ {\begin{array}{*{20}{c}}	{\Delta U'_{d}}\\	{\Delta U'_{q}} \end{array}} \right] \,,
\end{equation}
where $\Delta$ denotes the perturbed value of a variable, ${\bf Y_{PLL}}(s)$ is the $2 \times 2$ admittance matrix. For the detailed derivation and expression of ${\bf Y_{PLL}}(s)$ we refer to \cite{wen2016analysis, huang2019impacts}, etc.

\subsection{Admittance Modeling of Grid-Forming Converters}

Different from PLL-based converters, the current reference vector $\vec I^{\rm ref}$ of in the grid-forming controller in Fig.~\ref{Fig_converter_control} comes from the voltage loop as
\begin{equation}\label{eq:voltage_control}
\vec I^{\rm ref} = {\rm PI_{VC}}(s)(\vec V^{\rm ref} - \vec V) + j \omega C_F \vec V + k_F \vec I_g \,,
\end{equation}
where ${\rm PI_{VC}}(s)$ is the transfer function of the PI regulator, $\vec V^{\rm ref} = 1+j0$ is the voltage reference vector, and $k_F$ is the current feed-forward coefficient. Note that $\vec V^{\rm ref}$ can also be provided by a reactive power control loop if needed.

Moreover, the converter in Fig.~\ref{Fig_converter_control} achieves grid synchronization by emulating the swing equation as
\begin{equation}\label{eq:swing}
\left\{ \begin{array}{l}
s \theta = \omega \,, \\
Js\omega = P_0 - P_E - D(\omega - \omega_0) \,,
\end{array} \right.
\end{equation}
which determines the angle and frequency of the controller's {\em dq}-frame.

By combining the linearized form of \eqref{eq:Lf}-\eqref{eq:current_control}, \eqref{eq:PQ}, \eqref{eq:global_dq}, \eqref{eq:voltage_control}, and \eqref{eq:swing}, we derive the admittance model of grid-forming converters (VSMs in this paper) as
\begin{equation}\label{eq:Impedance_VSM}
\begin{array}{l}
- \left[ {\begin{array}{*{20}{c}}	{\Delta I'_{gd}}\\	{\Delta I'_{gq}} \end{array}} \right] = {\bf Y_{GF}}(s) \left[ {\begin{array}{*{20}{c}}	{\Delta U'_{d}}\\	{\Delta U'_{q}} \end{array}} \right] \,,\\
{\bf Y_{GF}}(s) = - \left[ {\begin{array}{*{20}{c}} Y(s) & 0 \\ \frac{Y^2(s)V^2_{d0} - (I'_{gd0})^2}{Js^2+Ds} & Y(s) \end{array}} \right] \,,
\end{array}
\end{equation}
where $\bf Y_{GF}(s)$ is the $2 \times 2$ admittance matrix, the subscript $_0$ denotes the steady-state value of a variable, and
\begin{equation}\label{eq:Ys}
\begin{array}{*{20}{l}}
Y(s) = {\frac{G_{\rm VF}(s)+G_{\rm CC}(s){\rm PI_{VC}}(s)+sC_F+j\omega C_F[1 - G_{\rm CC}(s)]}{k_F G_{\rm CC}(s)-1}} \,, \\
G_{\rm CC}(s) = \frac{{\rm PI_{CC}}(s)}{sL_F+{\rm PI_{CC}}(s)} \,, \\
G_{\rm VF}(s)=\frac{1-f_{\rm VF}(s)}{sL_F+{\rm PI_{CC}}(s)} \,.
\end{array}
\end{equation}
We note that generally $Y(s)$ has large magnitudes due to the voltage control, or in other words, the voltage control is supposed to have limited output impedance with appropriate control designs \cite{tao2015analysis}.

\subsection{Closed-Loop Dynamics of Multi-Converter Systems}

\begin{figure}[!t]
	\centering
	\includegraphics[width=3.4in]{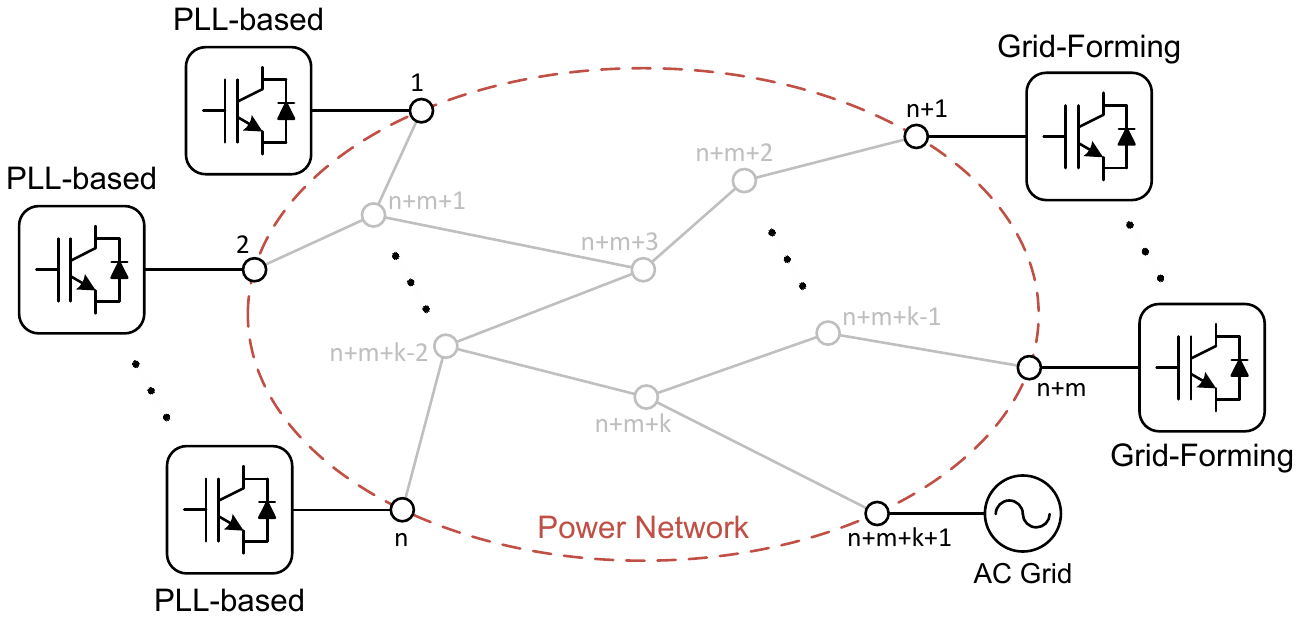}
	\vspace{-2mm}
	%\DeclareGraphicsExtensions.
	\caption{Illustration of a multi-converter system.}
	\vspace{-2mm}
	\label{Fig_network}
\end{figure}

Now we are ready to formulate the closed-loop dynamics of multi-converter systems. Consider a multi-converter system as depicted in Fig.~\ref{Fig_network} (the network topology is for illstration), which contains $n$ PLL-based converters (connected Nodes $1 \sim n$), $m$ grid-forming converters (connected to Nodes $n+1 \sim n+m$), $k$ interior nodes (Nodes $n+m+1 \sim n+m+k$), and an infinite bus (Node $n+m+k+1$). The interior nodes are not directly connected to the converters and will be eliminated through Kron reduction by assuming that the currents injected into these nodes remain constant \cite{dorfler2013kron}. The infinite bus can be considered an ``grounded'' node in small-signal modeling \cite{huang2019impacts}.

For a transmission line that connects Node $i$ and Node $j$, its dynamic equation can be expressed as
\begin{equation}\label{eq:nodeij}
\begin{split}
&\left[ {\begin{array}{*{20}{c}}
	{{\Delta I'_{d,ij}}}\\
	{{\Delta I'_{q,ij}}}
	\end{array}} \right]= {B_{ij}}F(s)\left[ {\begin{array}{*{20}{c}}
	{{\Delta U'_{d,i}} - {\Delta U'_{d,j}}}\\
	{{\Delta U'_{q,i}} - {\Delta U'_{q,j}}}
	\end{array}} \right]\,,\\
&F(s)= \frac{{{1}}}{{{(s+\tau)^2/\omega _0} + \omega _0}}\left[ {\begin{array}{*{20}{c}}
	{s+\tau}&{ {\omega _0}}\\
	{ - {\omega _0}}&{s+\tau}
	\end{array}} \right]\,,		
\end{split}
\end{equation}
where $\left[ {\begin{array}{*{20}{c}} 	{{\Delta I'_{d,ij}}}\\	{{\Delta I'_{q,ij}}} \end{array}} \right]$ is the current from $i$ to $j$ and $\left[ {\begin{array}{*{20}{c}} {{\Delta U'_{d,i}}}\\	{{\Delta U'_{q,i}}} \end{array}} \right]$ is the voltage at $i$ (in the global {\em dq}-frame), $B_{ij} = 1/(L_{ij}  \omega _0)$ is the susceptance between $i$ and $j$, and $\tau$ is the identical $R/L$ ratio of all the lines.

Let $Q \in \mathbb{R}^{(n+m+k)\times (n+m+k)}$ be the grounded Laplacian matrix of the electrical network which can be calculated by ${Q_{ij}} =  - {B_{ij}} (i \ne j)$ and ${Q_{ii}} = \sum\limits_{j = 1,i \ne j}^{n+m+k} {{B_{ij}}} + B_{i,n+m+k+1}$. By performing Kron reduction, we eliminate the interior nodes and obtained the Kron-reduced Laplacian matrix as
\begin{equation}\label{eq:Qred}
{Q_{{\rm{red}}}} = {Q_1} - {Q_2}  Q_4^{ - 1}  {Q_3}\,,		
\end{equation}
where ${Q_1} \in {\mathbb{R}^{(m+n) \times (m+n)}}$, ${Q_2} \in {\mathbb{R}^{(m+n) \times k}}$, ${Q_3} \in {\mathbb{R}^{k \times (m+n)}}$, ${Q_4} \in {\mathbb{R}^{k \times k}}$,
$Q = \left[ {\begin{array}{*{20}{c}}
	{{Q_1}}&\vline& {{Q_2}}\\
	\hline
	{{Q_3}}&\vline& {{Q_4}}
	\end{array}} \right]$.

Then, similar to \cite{huang2019impacts,huang2020h,dong2018small}, it can be deduced from \eqref{eq:nodeij} and \eqref{eq:Qred} that the {\em network dynamics} can be formulated as
\begin{equation}\label{eq:network_dynamics}
\Delta {\bf I'_g} = Q_{\rm red} \otimes F(s) \Delta {\bf U'_g} \,,
\end{equation}
where $\Delta {\bf I'_g} \in \mathbb{R}^{2m+2n}$ is the stacked current vector of the $n+m$ converters (in the global {\em dq}-frame) injected into the network, $\Delta {\bf U'_g} \in \mathbb{R}^{2m+2n}$ is the stacked voltage vector of the $n+m$ converters, and $\otimes$ denotes the Kronecker product.

On the other hand, considering \eqref{eq:Impedance_PLL} and \eqref{eq:Impedance_VSM}, the dynamics of the $n$ PLL-based converters and the $m$ grid-forming converters can be formulated as
\begin{equation}\label{eq:converter_dynamics}
-\Delta {\bf I'_g} = {\bf S} \otimes I_2 \left[ {\begin{array}{*{20}{c}}
	{{I_n \otimes {\bf Y_{PLL}}(s)}}& {{\bf 0}}\\
	{{\bf 0}}& {{I_m \otimes {\bf Y_{GF}}(s)}}
	\end{array}} \right] \Delta {\bf U'_g} \,,
\end{equation}
where $I_n \in \mathbb{R}^{n \times n}$ denotes the identity matrix, $\bf 0$ denotes the zero matrix with a proper dimension, ${\bf S} \in \mathbb{R}^{(n+m) \times (n+m)}$ is a diagonal matrix whose $i{\rm th}$ diagonal element ${\bf S}_{i}$ is the capacity ratio of the $i{\rm th}$ node's rated capacity to the base capacity of per-unit calculation (we use the same base values when performing per-unit calculations for all the converters). In the above formulation, for simplicity we ignore the power angle differences of the converters (which are generally small enough \cite{dong2018small}) such that the admittance matrices are the same when applying the same control scheme.

\begin{figure}[!t]
	\centering
	\includegraphics[width=2.8in]{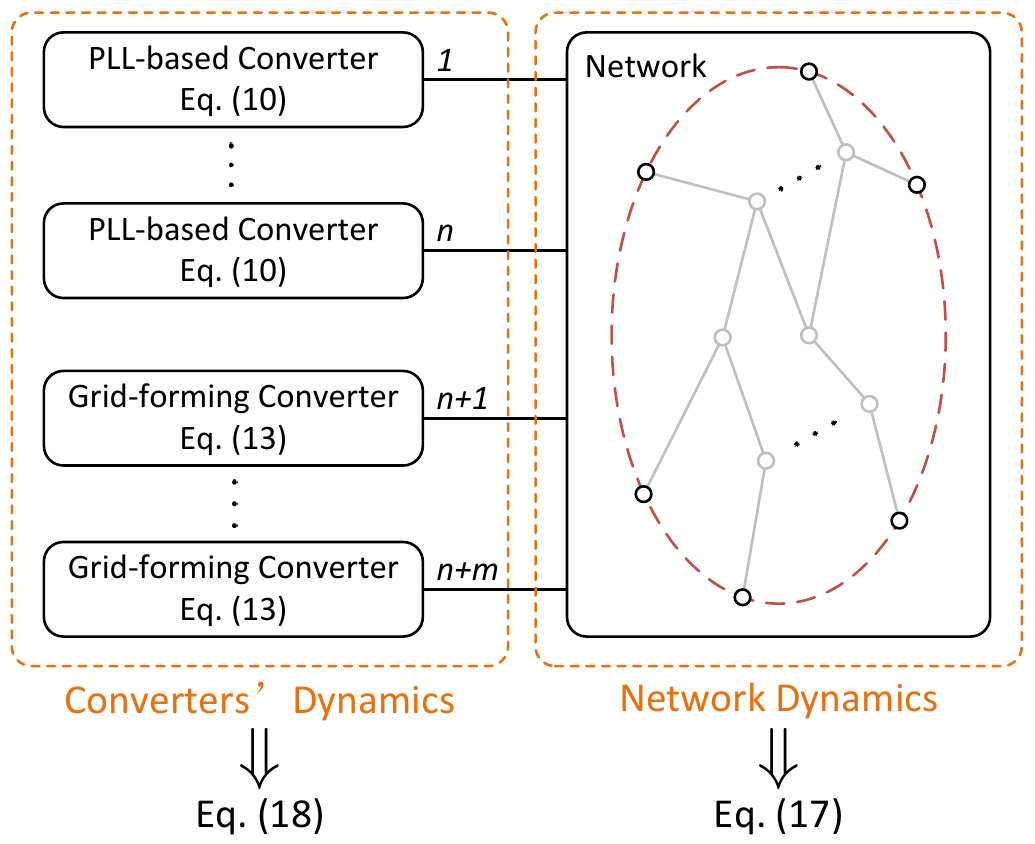}
	\vspace{-3mm}
	%\DeclareGraphicsExtensions.
	\caption{Modeling of the converters' dynamics and the network dynamics.}
	\vspace{0mm}
	\label{Fig_interaction}
\end{figure}

We remark that the above modeling is based on dividing the system into two parts, namely, the power network and the combination of all the converters, and then deriving the transfer function matrices of these two parts, as illustrated in Fig.~\ref{Fig_interaction}. The converters' dynamics are reflected in \eqref{eq:converter_dynamics}, which has a block diagonal structure and each block is either \eqref{eq:Impedance_PLL} or \eqref{eq:Impedance_VSM} as it represents the converter's dynamics. The network dynamics are reflected in \eqref{eq:network_dynamics}, which is derived based on the Kron-reduced Laplacian matrix of the power network, similar to the derivations in \cite{huang2019impacts} and \cite{huang2020h}.

\begin{figure}[!t]
	\centering
	\includegraphics[width=3.3in]{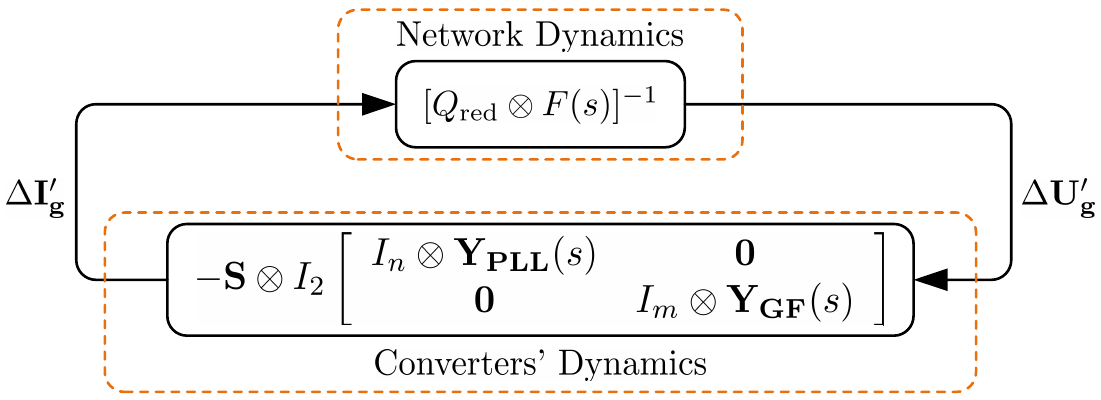}
	\vspace{-3mm}
	%\DeclareGraphicsExtensions.
	\caption{Closed-loop dynamics of multi-converter systems.}
	\vspace{-3mm}
	\label{Fig_closed_loop}
\end{figure}

By combining \eqref{eq:network_dynamics} and \eqref{eq:converter_dynamics}, we obtain the closed-loop dynamics of the multi-converter system as shown in Fig.~\ref{Fig_closed_loop}. We note that the block diagonal structure of the converters' dynamics facilitates the analysis of the impacts of grid-forming converters, which will be elaborated upon in the next section.

\section{Impacts of Grid-Forming Converters}
\label{sec:3}

Based on the above system modeling, in this section we focus on analyzing the impacts of grid-forming converters on the small-signal stability of the system integrated with PLL-based converters. For simplicity in the following analysis, we assume $m=1$, i.e., only one grid-forming converter is placed. We will show that the case of multiple grid-forming converters can be analyzed by repeating our analysis on $m=1$.

It can be deduced from Fig.~\ref{Fig_closed_loop} that the characteristic equation (with $m=1$) of the system is
\begin{equation}\label{eq:det_closed_loop}
\begin{split}
0&= {\rm det}\left( \left[ {\begin{array}{*{20}{c}}
	{{{\bf S_B} \otimes {\bf Y_{PLL}}(s)}}& {{\bf 0}}\\
	{{\bf 0}}& {{ {\bf Y_{GF}}(s)}}
	\end{array}} \right] + Q_{\rm red} \otimes F(s) \right) \\
&= {\rm det}\left[ {\bf C_{GF}}(s) \right] {\rm det}\{ {{\bf S_B} \otimes [{\bf Y_{PLL}}(s)F^{-1}(s)]} + Q_n \otimes I_2 \\
&\;\;\;\; - Q_{n,1} \otimes F(s) \; {\bf C}^{-1}_{\bf GF}(s) \; Q_{1,n} \otimes I_2 \} \; {\rm det}\left[ I_n \otimes F(s) \right] \,,
\end{split}
\end{equation}
where ${\rm det}(\cdot)$ denotes the determinant, ${\bf S_B}$ is the capacity ratio matrix for PLL converters (without loss of generality, we assume the capacity of the grid-forming converter equals the base capacity such that its capacity ratio is $1$, i.e., ${\bf S} = {\rm diag}\{{\bf S}_{\bf B},1\}$), ${\bf C_{GF}}(s) = {\bf Y_{GF}}(s) + Q_{n+1}F(s)$, $Q_n \in \mathbb{R}^{n\times n}$, $Q_{1,n} \in \mathbb{R}^{1\times n}$, $Q_{n,1} \in \mathbb{R}^{n\times 1}$, and $Q_{n+1} \in \mathbb{R}$ satisfy
\begin{equation}\label{eq:Q_n}
Q_{\rm red} = \left[ {\begin{array}{*{20}{c}}
	{Q_n}&\vline& {Q_{n,1}}\\
\hline
	{Q_{1,n}}&\vline& {Q_{n+1}}
	\end{array}} \right] \,.
\end{equation}

Notice that ${\rm det}[{\bf C_{GF}}(s)]$ is in fact the closed-loop characteristic equation of a grid-forming converter system connected to the infinite bus via a susceptance $Q_{n+1}$, which should be designed as stable. Moreover, as mentioned before, generally ${\bf Y_{GF}}(s)$ has large magnitudes due to the grid-forming design, so ${\bf C_{GF}}(s) \approx {\bf Y_{GF}}(s)$. Hence, it can be deduced from \eqref{eq:det_closed_loop} that we can instead focus on the following characteristic equation for evaluating the system stability
\begin{equation}\label{eq:det_closed_loop1}
\begin{split}
0={\rm det}\{ {I_n \otimes [{\bf Y_{PLL}}(s)F^{-1}(s)]} + ({\bf S}_{\bf B}^{-1} Q_n) \otimes I_2 \\
 - ({\bf S}_{\bf B}^{-1} Q_{n,1}Q_{1,n}) \otimes [F(s)  {\bf Y}^{-1}_{\bf GF}(s) ]  \} \,.
\end{split}
\end{equation}

According to the matrix perturbation theory \cite{stewart1990matrix}, \eqref{eq:det_closed_loop1} can be reformulated as
\begin{equation}\label{eq:det_perturb}
\begin{split}
0 =\;& [x \otimes a(s)]^\top \{I_n \otimes [{\bf Y_{PLL}}(s)F^{-1}(s)] + ({\bf S}_{\bf B}^{-1} Q_n) \otimes I_2 \\
&- ({\bf S}_{\bf B}^{-1}Q_{n,1}Q_{1,n}) \otimes [F(s)  {\bf Y}^{-1}_{\bf GF}(s) ] \} [y \otimes b(s)] \\
&+ o(\|F(s)  {\bf Y}^{-1}_{\bf GF}(s)\|^2) \,,
\end{split}
\end{equation}
where $x$ and $y$ are the left and right eigenvectors corresponding to the smallest eigenvalue of ${\bf S}_{\bf B}^{-1}Q_n$ (denoted by $\lambda'_1$) with normalization $x^\top y=1$, $a(s)$ and $b(s)$ are the normalized left and right eigenvectors of ${\bf Y_{PLL}}(s)F^{-1}(s)$ which satisfy $a^\top(s) [{\bf Y_{PLL}}(s)F^{-1}(s)] b(s) = \gamma(s)$ pertinent to the dominant poles of the system, $o(\|F(s)  {\bf Y}^{-1}_{\bf GF}(s)\|^2)$ is the second-order approximation error \cite[Theorem 2.3]{stewart1990matrix}. By ignoring this approximation error it can be further derived that
\begin{equation}\label{eq:det_perturb1}
\begin{split}
0 \approx \;& [x \otimes a(s)]^\top \{I_n \otimes [{\bf Y_{PLL}}(s)F^{-1}(s)] + ({\bf S}_{\bf B}^{-1}Q_n) \otimes I_2 \\
&- ({\bf S}_{\bf B}^{-1}Q_{n,1}Q_{1,n}) \otimes [F(s)  {\bf Y}^{-1}_{\bf GF}(s) ] \} [y \otimes b(s)] \\
= \; & [a^\top(s){\bf Y_{PLL}}(s)F^{-1}(s)b(s)] + (x^\top {\bf S}_{\bf B}^{-1}Q_n y)\\
&- (x^\top {\bf S}_{\bf B}^{-1} Q_{n,1}Q_{1,n} y)[a^\top(s)F(s)  {\bf Y}^{-1}_{\bf GF}(s)b(s)] \\
= \; & \gamma(s) + \lambda'_1 + \Delta \lambda(s) \,,
\end{split}
\end{equation}
where
\begin{equation}\label{eq:Delta_lambda}
\Delta \lambda(s) = - (x^\top {\bf S}_{\bf B}^{-1} Q_{n,1}Q_{1,n} y)[a^\top(s)F(s)  {\bf Y}^{-1}_{\bf GF}(s)b(s)] \,.
\end{equation}

Eq.~\eqref{eq:det_perturb1} describes how the grid-forming converter interacts with PLL-based converters via the power network and thus affects the dominant poles of the whole system. In the following, we introduce two propositions in order to provide a intelligible interpretation on \eqref{eq:det_perturb1}.

\begin{proposition}[Single Converter System]
\label{prop:Single_Converter_system}
The dominant poles of a PLL-based converter that is connected to an infinite bus (with line susceptance being $\lambda_1$) are determined by
\begin{equation}\label{eq:PLL_det}
0 = \gamma(s) + \lambda_1 \,,
\end{equation}
where the definition of $\gamma(s)$ has been given above.
\end{proposition}
\begin{proof}
It can be deduced from \eqref{eq:network_dynamics} and \eqref{eq:converter_dynamics} that the characteristic equation of such a system ($n=1,m=0$) is
\begin{equation}\label{PLL_det1}
\begin{split}
0 = \; &{\rm det}[{\bf Y_{PLL}}(s)+\lambda_1F(s)] \\
= \; &  {\rm det}[{\bf Y_{PLL}}(s)F^{-1}(s) + \lambda_1 I_2]{\rm det}[F(s)] \,.
\end{split}
\end{equation}

According to the definitions of $a(s)$ and $b(s)$ as given above, the dominant poles can be obtained by solving
\begin{equation}\label{PLL_det2}
0 = a^\top(s)[{\bf Y_{PLL}}(s)F^{-1}(s) + \lambda_1 I_2]b(s)
= \gamma(s) + \lambda_1 \,,
\end{equation}
which concludes the proof.
\end{proof}

The following proposition will show that \eqref{eq:PLL_det} also determines the dominant poles of multi-PLL-based-converter systems under certain circumstances.

\begin{proposition}[Multi-Converter system \cite{dong2018small}]
\label{prop:Multi_Converter_system}
Consider a power network (with the Kron-reduced Laplacian matrix $Q_{\rm red} \in \mathbb{R}^{p \times p}$) which interconnects $p$ ($p \in \mathbb{Z}_+$) PLL-based converters (with the capacity ratio matrix being ${\bf S}$). The dominant poles of this system can be obtained by solving \eqref{eq:PLL_det} (wherein $\lambda_1$ is the smallest eigenvalue of ${\bf S}^{-1}Q_{\rm red}$, also defined as the generalized short-circuit ratio (gSCR) in \cite{dong2018small} to evaluate the power grid strength).
\end{proposition}
\begin{proof}
It can be deduced from \eqref{eq:network_dynamics} and \eqref{eq:converter_dynamics} that the characteristic equation of such a system ($m=0$) is
\begin{equation}\label{eq:multi_PLL_det}
0 = {\rm det}[I_p \otimes {\bf Y_{PLL}}(s) + ({\bf S}^{-1}Q_{\rm red}) \otimes F(s)] \,,
\end{equation}
which is equivalent to
\begin{equation}\label{eq:multi_PLL_det1}
\begin{split}
0 = \; & {\rm det}\{(T^{-1}\otimes I_2)[I_p \otimes {\bf Y_{PLL}}(s) \\
\; & \;\;\;\;\;\;\;+ {\bf S}^{-1}Q_{\rm red} \otimes F(s)](T \otimes I_2)\} \\
= \; & {\rm det}[I_p \otimes {\bf Y_{PLL}}(s) + \Lambda_p \otimes F(s)] \\
= \; & \prod\limits_{i=1}^{p} {{\rm det}[{\bf Y_{PLL}}(s) + \lambda_i F(s)]} \,,
\end{split}
\end{equation}
where $T \in \mathbb{R}^{p \times p}$ diagonalizes ${\bf S}^{-1}Q_{\rm red}$ as $T^{-1}{\bf S}^{-1}Q_{\rm red}T = \Lambda_p = {\rm diag}\{\lambda_1,\lambda_2,...,\lambda_p\}$ ($\lambda_1<\lambda_2<...<\lambda_p$).

Hence, the system that has $p$ PLL-based converters can be decoupled into $p$ subsystems as indicated by \eqref{eq:multi_PLL_det1}. Moreover, it can be deduced that the dominant poles is determined by the weakest system, which is in fact a single PLL-based converter connected to an infinite bus with susceptance $\lambda_1$ \cite{dong2018small, huang2019impacts}.
Then, by Proposition~\ref{prop:Single_Converter_system}, the dominant poles of this weakest system are determined by \eqref{eq:PLL_det}, which concludes the proof.
\end{proof}

\begin{remark}[Power Grid Strength]\label{Remark1}
The value of $\lambda_1$ in \eqref{eq:PLL_det} reflects the network connectivity and thus the power grid strength, also defined in \cite{dong2018small} as the generalized short-circuit ratio (gSCR) of the system. Moreover, $\lambda_1$ determines the stability margin of PLL-based multi-converter systems (a larger $\lambda_1$ indicates a larger stability margin), and a low $\lambda_1$ may give rise to PLL-induced instabilities \cite{huang2019impacts}.
\end{remark}

Based on the above results, it can be deduced by comparing \eqref{eq:det_perturb1} and \eqref{eq:PLL_det} that placing grid-forming converters is equivalent to the change of power grid strength (i.e., gSCR) from $\lambda_1$ to $\lambda'_1 + \Delta \lambda (s)$. We note that when the capacity of the grid-forming converter is sufficiently large (i.e., $\lambda'_1$ is relatively large in such cases), $\Delta \lambda (s)$ in \eqref{eq:Delta_lambda} has magnitudes small enough to be ignored in the frequency range of control effects because ${\bf Y_{GF}}(s)$ has large magnitudes due to the voltage control. We will verify this issue when conducting case studies. Hence, we have the following statement.

\begin{remark}[Grid-Forming Control and Power Grid Strength]
\label{Remark2}
The main impact of placing grid-forming converters (i.e., changing the control scheme of a converter from PLL-based control to grid-forming control) can be interpreted as changing the power grid strength (i.e., gSCR) from $\lambda_1$ to $\lambda'_1$.
\end{remark}

For clarity, we recall that we consider a multi-converter system which contains $n$ PLL-based converters and one grid-forming converter as in \eqref{eq:det_closed_loop}, and $\lambda_1$ is the smallest eigenvalue of ${\bf S}^{-1}Q_{\rm red} \in \mathbb{R}^{(n+1) \times (n+1)}$ while $\lambda'_1$ is the smallest eigenvalue of ${\bf S}_{\bf B}^{-1}Q_n$ ($Q_n$ is defined in \eqref{eq:Q_n}). Note that according to \cite[Lemma 1]{dong2018small}, there holds $\lambda'_1 >0$ and $\lambda_1>0$. The following lemma is given to compare $\lambda'_1$ with $\lambda_1$.

\begin{lemma}
\label{Lemma_lambda}
Consider two weighted Kron-reduced Laplacian matrices ${\bf S}^{-1}Q_{\rm red} \in \mathbb{R}^{(n+1) \times (n+1)}$ and ${\bf S}_{\bf B}^{-1}Q_n \in \mathbb{R}^{n \times n}$, where ${\bf S}_{\bf B}^{-1}Q_n$ is a submatrix of ${\bf S}^{-1}Q_{\rm red}$ considering \eqref{eq:Q_n} and ${\bf S} = {\rm diag}\{{\bf S}_{\bf B},1\}$. It holds that $\lambda_1 < \lambda'_1$, where $\lambda_1$ and $\lambda'_1$ are respectively the smallest eigenvalues of ${\bf S}^{-1}Q_{\rm red}$ and ${\bf S}_{\bf B}^{-1}Q_n$.
\end{lemma}
\begin{proof}
The claimed result can be obtained by considering the interlacing theorem in \cite[Theorem 4.3.8]{horn2012matrix}.
\end{proof}

With Remark~\ref{Remark2} and Lemma~\ref{Lemma_lambda}, it can be further deduced that the placement of a grid-forming converter is equivalent to increasing the power network strength (characterized by gSCR) from $\lambda_1$ to $\lambda'_1$. Moreover, since $\lambda'_1$ is the smallest eigenvalue obtained by deleting the $(n+1) {\rm th}$ row and $(n+1) {\rm th}$ column of ${\bf S}^{-1}Q_{\rm red}$, the case of multiple grid-forming converters can be analyzed by repeating the above analysis and calculating the smallest eigenvalue of ${\bf S}^{-1}Q_{\rm red}$ after deleting the rows and columns corresponding to the nodes of grid-forming converters. Note that if the grid-forming converter is labelled as the $i{\rm th}$ converter in \eqref{eq:det_closed_loop}, then $\lambda'_1$ should be the smallest eigenvalue obtained by deleting the $i {\rm th}$ row and $i {\rm th}$ column of ${\bf S}^{-1}Q_{\rm red}$.
We summarize the above finding in the following proposition.

\begin{proposition}[Placement of Grid-Forming Converters]
\label{prop:Grid_Forming}
Consider a multi-converter system whose closed-loop dynamics are described by Fig.~\ref{Fig_closed_loop} (the number of grid-forming converters could be $0$), the PLL-induced small-signal stability of the system is improved by changing the control scheme of any PLL-based converter to grid-forming control.
\end{proposition}
\begin{proof}
The above statement rigorously follows from Remark~\ref{Remark1}, Remark~\ref{Remark2} and Lemma~\ref{Lemma_lambda}.
\end{proof}

So far, we have theoretically explained how the placement of grid-forming converters enhances the power grid strength and thus the (PLL-induced) small-signal stability. One remaining question is how to optimally place the grid-forming converters to improve the stability, which will be explored in the following section.

\section{Optimal Placement of Grid-Forming Converters}
\label{sec:4}

This section investigates the problem of how to optimally place grid-forming converters to improve the system stability. To be specific, we consider a power system that is integrated with $p$ PLL-based converters, and $q$ ($q < p$) of them will be changed to use grid-forming control instead of PLL-based control. Define a symmetric weighted Laplacian matrix $\mathcal L = {\bf S}^{-\frac{1}{2}}Q_{\rm red}{\bf S}^{-\frac{1}{2}}$, which shares the same eigenvalues with ${\bf S}^{-1}Q_{\rm red}$ because they are similar matrices.
According to the analysis and results in the previous section, the optimal locations to place these grid-forming converters can be obtained by solving the following optimization problem
\begin{equation}\label{eq:optimal_placement}
\max \limits_{\mathcal{I} \subset \mathcal{V}}\;\;\lambda_{\min}[\mathcal{R}_{\mathcal{V} \backslash \mathcal{I}}(\mathcal{L})] \,,
\end{equation}
where $\mathcal{V} = \{1,2,...,p\}$ is the set that denotes the converter nodes, $\mathcal{I} \subset \mathcal{V}$ is the set of the $q$ nodes to place the grid-forming converters which will be obtained by solving the optimization problem, $\lambda_{\min}(\cdot)$ denotes the smallest eigenvalue of a symmetric matrix, $\mathcal{L} \in \mathbb{R}^{p \times p}$ is the symmetric weighted Laplacian matrix to represent the power network and the capacities of the converters, and $\mathcal{R}_{\mathcal{V} \backslash \mathcal{I}}(\mathcal{L})$ denotes the remaining matrix after deleting the rows and columns included in the set $\mathcal{I}$. That is to say, \eqref{eq:optimal_placement} aims at selecting the locations of grid-forming converters to equivalently enhance the power grid strength (i.e., gSCR) and thus the system stability.

Since \eqref{eq:optimal_placement} is a combinatorial optimization problem, one can solve it by simply enumerating all possible results, which is doable if the system is small-scale. However, if the system is integrated with large-scale converters, the computational burden of the enumeration will be unacceptable. As a remedy, we propose a greedy method to obtain a suboptimal solution for the placement of grid-forming converters, and we will show by case studies that this suboptimal solution is identical to the optimal solution in some cases and has very satisfactory performance to be used in practice. To be specific, we consider now the following iterative optimization problem that will be solved for $q$ times
\begin{equation}\label{eq:optimal_placement1}
\max \limits_{\alpha_i \in \mathcal{V}}\;\;\lambda_{\min}[\mathcal{R}_{\mathcal{V} \backslash \alpha_i}(\mathcal{L}^{[i]})] \,,
\end{equation}
where $i \in \{1,2,...,q\}$ denotes the iteration number, $\alpha_i \in \mathcal{V}$ is the node to place a grid-forming converter that will be determined by the $i{\rm th}$ iteration, $\mathcal{R}_{\mathcal{V} \backslash \alpha_i}(\cdot)$ is a function to delete the row and the column that are corresponding to the Node~$\alpha_i$ defined in $\mathcal{L}$, $\mathcal{L}^{[i+1]} = \mathcal{R}_{\mathcal{V} \backslash \alpha_i}(\mathcal{L}^{[i]})$, with $\mathcal{L}^{[1]} = \mathcal{L} \in \mathbb{R}^{p \times p}$. Hence, by iteratively solving \eqref{eq:optimal_placement1} for $q$ times using a greedy heuristic, the locations for placing the $q$ grid-forming converters can be obtained. We summarize the solving process in the following.

\begin{algorithm}[h]
\caption{\normalsize Greedy Solution to \eqref{eq:optimal_placement}} \label{Algorithm:greedy}
{\bf{Input:}} Kron-reduced Laplacian matrix: $Q_{\rm red} \in \mathbb{R}^{p \times p}$, capacity ratio matrix: ${\bf S} \in \mathbb{R}^{p \times p}$, the number of grid-forming converters to be placed: $q$
	
\begin{enumerate}[ 1)]
	\item Initialize $i=1$, $\mathcal{L}^{[i]}|_{i=1} = \mathcal{L} = {\bf S}^{-\frac{1}{2}}Q_{\rm red}{\bf S}^{-\frac{1}{2}}$.
	\item Solve \eqref{eq:optimal_placement1} for $\alpha_i$.
	\item Calculate $\mathcal{L}^{[i+1]} = \mathcal{R}_{\mathcal{V} \backslash \alpha_i}(\mathcal{L}^{[i]})$.
	\item Iterate through Steps 2-3 for $i = 2,3,...,q$.
\end{enumerate}

{\bf{Output:}} The (sub)optimal locations for the $q$ grid-forming converters, i.e., $\alpha_i$ ($i \in \{1,2,...,q\}$).

\end{algorithm}

In Step~2, a trivial way to obtain the solution of \eqref{eq:optimal_placement1} (without loss of generality we assume $i=1$ here) is to enumerate the smallest eigenvalues of $\mathcal{R}_{\mathcal{V} \backslash \alpha_1}(\mathcal{L})$ with $\alpha_1 = 1,2,...,p$ and then pick the largest one. This method requires to do eigenvalue calculation for $p$ times, which may also have high computational burden in large-scale systems. In fact, deleting the $\alpha_1 {\rm th}$ row and the $\alpha_1 {\rm th}$ column can be interpreted as connecting Node~$\alpha_1$ directly to the grounded node (i.e., the infinite bus). Hence, the solution of \eqref{eq:optimal_placement1} can be considered as the ``farthest'' node from the grounded node, considering that connecting the farthest node to the grounded node will increase the network connectivity (reflected by the smallest eigenvalue \cite{dorfler2013kron,dorfler2018electrical}) to the most extent.

This farthest node can be approximately located by checking the participation factors of the nodes on the smallest eigenvalue of $\mathcal{L}$ (i.e., $\lambda_1$). According to \cite{huang2019impacts}, the participation factor of Node~$i$ on $\lambda_1$ equals the sensitivity of $\lambda_1$ to the self-susceptance of Node~$i$ (included in $\mathcal{L}_{ii}$), formulated as
\begin{equation}\label{eq:Q_red_sensitivity}
p_{1,i} = \frac{\partial \lambda_1}{\partial \mathcal{L}_{ii}} = v_{1,i}u_{1,i}
\end{equation}
where $p_{1,i}$ denotes the participation factor of Node~$i$ on $\lambda_1$, $\mathcal{L}_{ii}$ denotes the $i{\rm th}$ diagonal entry of $\mathcal{L}$, $v_1 \in \mathbb{R}^{p}$ and $u_1 \in \mathbb{R}^{p}$ are respectively the left and right eigenvectors corresponding to $\lambda_1$ (i.e., $v_1^\top \mathcal{L}u_1 = \lambda_1$), $v_{1,i}$ and $u_{1,i}$ are the $i{\rm th}$ entries of $v_1$ and $u_1$.
Then, we pick the node that has the largest participation factor among all the $p$ nodes, which can be considered as an suboptimal solution to \eqref{eq:optimal_placement1} (we will showcase that this suboptimal solution leads to very satisfactory results and in some case it would be identical to the optimal solution). Note that this method only requires to do the eigendecomposition once, which is much more efficient than the enumeration method.

\section{Simulation Results}

In this section, we provide detailed simulation results to illustrate the validity of our analysis and the effectiveness of the proposed algorithm for the optimal placement of grid-forming converters.

\subsection{Case studies on a four-converter test system}

Without loss of generality, we consider a two-area four-converter test system as shown in Fig.~\ref{Fig_2Area}. Also note that our analysis and conclusions in this paper are general to other multi-converter systems with arbitrary network topology. In Fig.~\ref{Fig_2Area}, where Nodes~$1 \sim 4$ are converter nodes, Nodes~$5 \sim 9$ are interior nodes (which will be eliminated via Kron reduction as in \eqref{eq:Qred}), and Node~$10$ is the infinite bus (considered as the grounded node in the small-signal modeling). For simplicity, the capacities of the converters are assumed to be the same, i.e., $\bf S$ is an identity matrix and thus $\mathcal{L} = Q_{\rm red}$.
The main parameters of the systems are given in the Appendix, and the Kron-reduced Laplacian matrix is calculated as
\begin{equation*}
Q_{\rm red} = \mathcal{L} = \left[\scriptsize {\begin{array}{*{20}{r}}
	{8.07}& {-3.86}& {-0.20} & {-0.27} \\
	{-3.86}&{12.27}& {-0.41} & {-0.54}  \\
	{-0.20}&{-0.41}& {4.04}  & {-1.27} \\
	{-0.27}&{-0.54}& {-1.27} & {4.9675} \\	
\end{array}} \right] \,,
\end{equation*}
with the smallest eigenvalue being $\lambda_1 = 3.00$.

\begin{figure}[!t]
	\centering
	\includegraphics[width=3.2in]{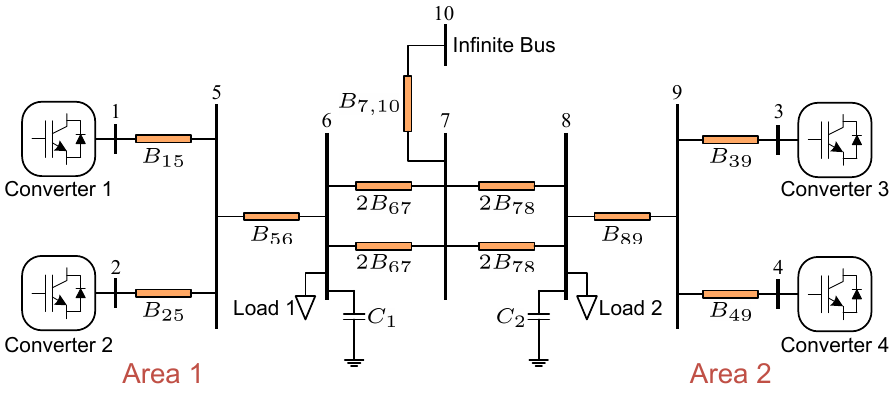}
	\vspace{-3mm}
	%\DeclareGraphicsExtensions.
	\caption{A two-area test system.}
	\vspace{0mm}
	\label{Fig_2Area}
\end{figure}

To begin with, we give the Bode diagram of $\Delta \lambda (s)$ (as defined in \eqref{eq:det_perturb1} and \eqref{eq:Delta_lambda}) to verify our claim in Section~\ref{sec:3} that the magnitudes of $\Delta \lambda (s)$ are small enough to be ignored. It can be seen from Fig.~\ref{Fig_Bode_Delta_Lambda} that in the frequency range of interest (within $1{\rm Hz} \sim 200{\rm Hz}$ regarding PLL-induced instabilities), the magnitudes of $\Delta \lambda (s)$ are around $0.01$, which are small enough to be ignored. Hence, $\Delta \lambda (s)$ can be ignored when analyzing PLL-induced instabilities, which leads to the claimed Remark~\ref{Remark2} and the subsequent results in Section~\ref{sec:4}.

\begin{figure}[!t]
	\centering
	\includegraphics[width=2.5in]{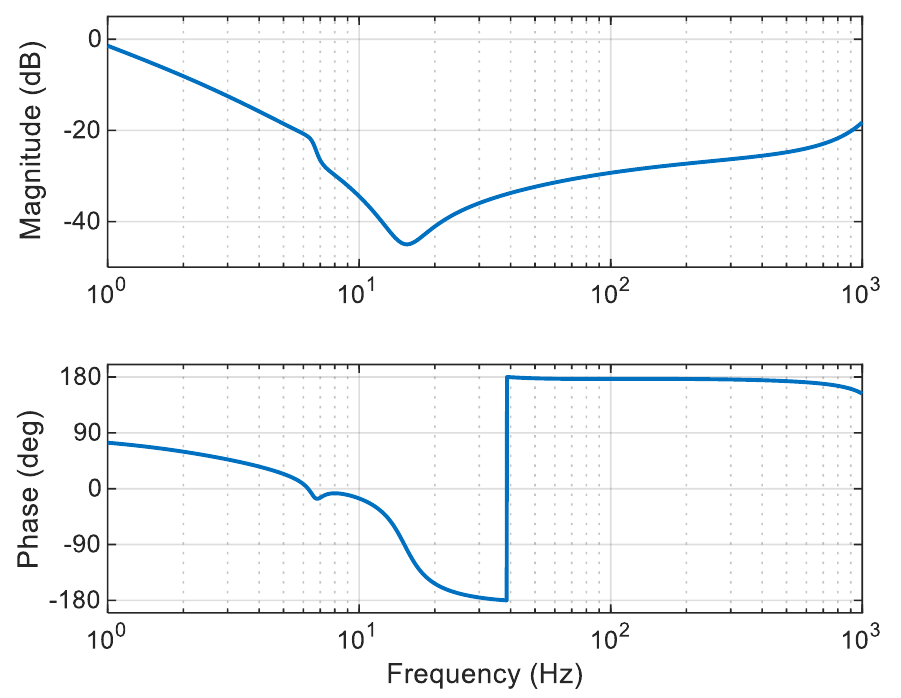}
	\vspace{-3mm}
	%\DeclareGraphicsExtensions.
	\caption{The Bode diagram of $\Delta \lambda (s)$ with the parameters in the Appendix.}
	\vspace{-3mm}
	\label{Fig_Bode_Delta_Lambda}
\end{figure}

Based on the proposed method in Section~\ref{sec:4}, we explore now the optimal placement of grid-forming converters in the two-area test system. Assume that two (out of four) converters will be changed from PLL-based control to grid-forming control, i.e., $q=2$. According to the analysis in Section~\ref{sec:4}, the optimal locations of these two grid-forming converters can be obtained by solving combinatorial optimization problem in \eqref{eq:optimal_placement}. In Table~\ref{table:lambda_min}, we enumerate all the possible combinations. Obviously, the optimal locations are $\{3,4\}$, that is, by changing the control schemes of Converter~3 and Converter~4 from PLL-based control to grid-forming control, the (PLL-induced) stability margin will be increased to the most extent.

\vspace{-2mm}

\renewcommand\arraystretch{1.45}
\begin{table}[h!]
	\scriptsize
	\centering
	\caption{The value of $\lambda_{\min}[\mathcal{R}_{\mathcal{V} \backslash \mathcal{I}}(\mathcal{L})]$ with different $\mathcal{I}$}
    \vspace{-1mm}
	%\begin{tabular}{|c|c|}
%		\hline
%%		\multicolumn{2}{|c|}{Base values for per-unit calculation}										\\
%%		\hline
%        $\mathcal{I}$ & $\lambda_{\min}[\mathcal{R}_{\mathcal{V} \backslash \mathcal{I}}(\mathcal{L})]$\\
%		\hline
%        $\{1,2\}$ & {2.5191} \\
%        $\{1,3\}$ & {5.7214} \\
%        $\{1,4\}$ & {3.9317} \\
%        $\{2,3\}$ & {5.7167} \\
%        $\{2,4\}$ & {3.9300} \\
%        $\{3,4\}$ & {6.9879} \\
%        \hline
%	\end{tabular}	
    \begin{tabular}{|c|c|c|c|c|c|c|}	
    \hline
    $\mathcal{I}$ & $\{1,2\}$ & $\{1,3\}$ & $\{1,4\}$ & $\{2,3\}$ & $\{2,4\}$ & $\{3,4\}$  \\
    \hline
    {\hspace{-2.5mm} $\lambda_{\min}[\mathcal{R}_{\mathcal{V} \backslash \mathcal{I}}(\mathcal{L})]$ \hspace{-2.5mm}} & {3.1504} & {4.9270} & {4.0239 } & {4.9437} & {4.0338} & {5.7711} \\
    \hline
	\end{tabular}
	\vspace{-1mm}
	\label{table:lambda_min}
	\end{table}

On the other hand, solving \eqref{eq:optimal_placement} through enumeration will result in unacceptable computational burden in large-scale systems. To alleviate this problem, the greedy solution obtained by Algorithm~\ref{Algorithm:greedy} can be used.

In the following, we still consider the two-area system in Fig.~\ref{Fig_2Area} and $q=2$, and use Algorithm~\ref{Algorithm:greedy} to obtain the two optimal locations. For a first step, we solve \eqref{eq:optimal_placement1} for the first location $\alpha_1$ (i.e., $i=1$ and $\mathcal{L}^{[1]} = \mathcal{L}$). Table~\ref{table:lambda_min1} enumerates all the possible results, and obviously, the solution is $\alpha_1 = 3$.
As we discussed in Section~\ref{sec:4}, another convenient (and computationally efficient) way to obtain $\alpha_1$ is to check the participation factors, as shown in \eqref{eq:Q_red_sensitivity}. Table~\ref{table:p_i1} shows the participation factors of the converter nodes on the smallest eigenvalue of $\mathcal{L}^{[1]}$. It can be seen that Node~3 has the largest participation factor and thus $\alpha_1 = 3$, which is consistent with the enumeration results in Table~\ref{table:lambda_min1}.

\renewcommand\arraystretch{1.45}
\begin{table}[h!]
	\scriptsize
	\centering
	\caption{\vspace{1mm} The value of $\lambda_{\min}[\mathcal{R}_{\mathcal{V} \backslash \alpha_1}(\mathcal{L}^{[1]})]$ with different $\alpha_1$}
	\vspace{-2mm}
	\begin{tabular}{|c|c|c|c|c|}
		\hline
%		\multicolumn{2}{|c|}{Base values for per-unit calculation}										\\
%		\hline
        $\alpha_1$ & $1$ & $2$ & $3$ & $4$ \\
        \hline
        $\lambda_{\min}[\mathcal{R}_{\mathcal{V} \backslash \alpha_1}(\mathcal{L}^{[1]})]$ & {3.1046} & {3.1292} & {4.7091} & {3.9570} \\
        \hline
	\end{tabular}	
	\vspace{-6mm}
	\label{table:lambda_min1}
\end{table}

\renewcommand\arraystretch{1.45}
\begin{table}[h!]
	\scriptsize
	\centering
	\caption{Participation factors of $\mathcal{L}^{[1]}$}
	\vspace{-2mm}
	\begin{tabular}{|c|c|c|c|c|}
		\hline
%		\multicolumn{2}{|c|}{Base values for per-unit calculation}										\\
%		\hline
        $\alpha_1$ & $1$ & $2$ & $3$ & $4$ \\
        \hline
        Participation factor & {0.0284} & {0.0193} & {0.6231} & {0.3292} \\
        \hline
	\end{tabular}		
	\vspace{-1mm}
	\label{table:p_i1}
\end{table}

After selecting the first optimal location $\alpha_1$, we have
\begin{equation*}
\mathcal{L}^{[2]} = \mathcal{R}_{\mathcal{V} \backslash \alpha_1}(\mathcal{L}^{[1]}) = \left[\scriptsize {\begin{array}{*{20}{r}}
{8.07}& {-3.86} & {-0.27} \\
{-3.86}&{12.27} & {-0.54}  \\
{-0.27}&{-0.54} & {4.9675} \\	
\end{array}} \right] \,.
\end{equation*}
Note that the row~3 and column~3 of the above matrix are corresponding to the Converter~4 in Fig.~\ref{Fig_2Area}.
Then, we solve \eqref{eq:optimal_placement1} for the second location $\alpha_2$. Table~\ref{table:lambda_min2} enumerates all the remaining converter nodes, and it can be seen that the optimal solution is $\alpha_2 = 4$. To illustrate that this optimal location can also simply obtained by checking the participation factors of the remaining converters nodes on the smallest eigenvalue, Table~\ref{table:p_i2} gives the participation factors of the Nodes~1, 2 and 4, which indicates that Converter~4 has the largest participation factor and thus the second optimal location is $\alpha_2 = 4$, consistent with the result obtained in Table~\ref{table:lambda_min2}.

%\vspace{-2mm}

\renewcommand\arraystretch{1.45}
\begin{table}[h!]
	\scriptsize
	\centering
	\caption{\vspace{1mm} The value of $\lambda_{\min}[\mathcal{R}_{\mathcal{V} \backslash \alpha_2}(\mathcal{L}^{[2]})]$ with different $\alpha_2$}
	\vspace{-2mm}
	\begin{tabular}{|c|c|c|c|}
		\hline
%		\multicolumn{2}{|c|}{Base values for per-unit calculation}										\\
%		\hline
        $\alpha_2$ & $1$ & $2$ & $4$ \\
        \hline
        $\lambda_{\min}[\mathcal{R}_{\mathcal{V} \backslash \alpha_2}(\mathcal{L}^{[2]})]$ & {4.9270} & {4.9437} & {5.7711} \\
        \hline
	\end{tabular}		
	\vspace{-2mm}
	\label{table:lambda_min2}
\end{table}

\renewcommand\arraystretch{1.45}
\begin{table}[h!]
	\scriptsize
	\centering
	\caption{Participation factors of $\mathcal{L}^{[2]}$}
	\vspace{-3mm}
	\begin{tabular}{|c|c|c|c|}
		\hline
%		\multicolumn{2}{|c|}{Base values for per-unit calculation}										\\
%		\hline
        $\alpha_2$ & $1$ & $2$ & $4$ \\
        \hline
        Participation factor & {0.1283} & {0.0614} & {0.8103} \\
        \hline
	\end{tabular}	
	\vspace{2mm}
	\label{table:p_i2}
\end{table}

It can be seen from the above results that the greedy solution obtained by employing Algorithm~\ref{Algorithm:greedy} (i.e., $\alpha_1 = 3$ and $\alpha_2 = 4$) is fully consistent with the optimal solution by directly solving \eqref{eq:optimal_placement} (i.e., $\mathcal{I} = \{3,4\}$), which verifies the effectiveness of the greedy algorithm. Although the greedy algorihm is not guaranteed to reach the optimal solution of \eqref{eq:optimal_placement}, i.e., sometimes it would lead to a suboptimal solution, it can generally satisfy the practical expectation and the obtained locations to place the grid-forming converters can significantly improve the system stability. We consider the development of rigorously optimal and computationally efficient solutions to \eqref{eq:optimal_placement} as future works.

In the following, we provide time-domain simulation results and eigenvalue analysis on the two-area test system in Fig.~\ref{Fig_2Area} to further verify the effectiveness of the optimal placement of grid-forming converters. Fig.~\ref{Fig_sim} shows the converters' responses under different control settings (an overload event occurs at 0.2s and is cleared after 0.02s). When all the four converters apply PLL-based control (Case~1), the system is oscillating (critically stable) which is caused by the interaction between PLLs and weak grid condictions. When applying grid-forming control in Converters~1 and 2 (Case~2), the damping ratio is improved and the oscillation is restrained, that is, the placements of grid-forming converters improve the system stability. Moreover, by applying grid-forming control in Converters~3 and 4 (Case~3), the system has very high damping ratio, indicating a satisfactory stability margin, which is consistent with the above results on the optimal placements of grid-forming converters.

\begin{figure}[!t]
	\centering
	\includegraphics[width=3.4in]{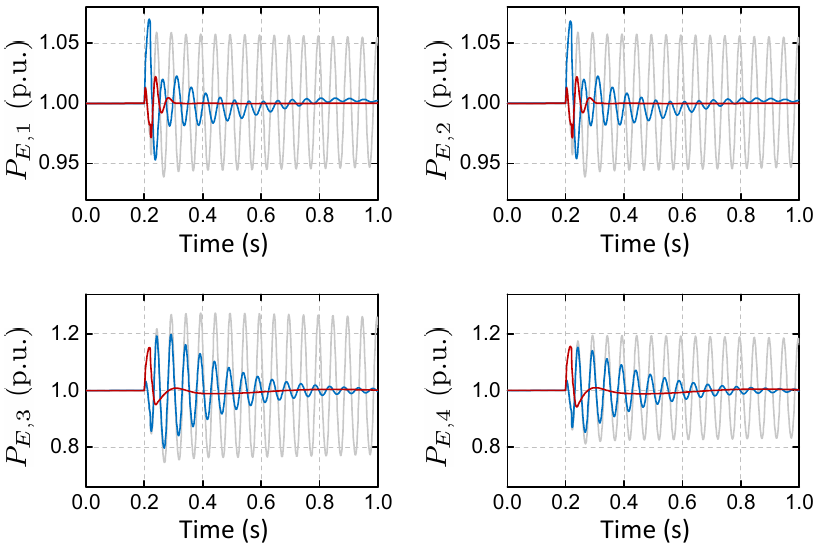}
	\vspace{-2mm}
	%\DeclareGraphicsExtensions.
	\caption{Time-domain responses of the four converters when placing grid-forming controls at different spots (an overload event occurs at 0.2s and is cleared after 0.02s). {\color{CGREY} \bf ---} : Case~1 (All the converters use PLL-based control); {\color{CBLUE} \bf ---} : Case~2 (Grid-forming control is applied in Converters~1 and 2); {\color{CRED} \bf ---} : Case~3 (Grid-forming control is applied in Converters~3 and 4).}
	\vspace{0mm}
	\label{Fig_sim}
\end{figure}

%Table~\ref{table:poles} gives the dominant poles of the system corresponding to the above control settings (Cases~$1 \sim 3$), which also shows that the proper placements of grid-forming converters significantly improve the system stability.
%
%\renewcommand\arraystretch{1.45}
%\begin{table}[h!]
%	\scriptsize
%	\centering
%	\caption{Dominant poles of the three cases}
%    \vspace{-1mm}
%	\begin{tabular}{|c|c|c|c|}
%		\hline
%%		\multicolumn{2}{|c|}{Base values for per-unit calculation}										\\
%%		\hline
%        {} & Case~1 & Case~2 & Case~3 \\
%        \hline
%        Dominant poles & {$-1.0 \pm j126.5$} & {$-4.8 \pm j127.7$} & {$-58.7 \pm j136.7$} \\
%        \hline
%	\end{tabular}		
%	\vspace{-3mm}
%	\label{table:poles}
%\end{table}

Fig.~\ref{Fig_Eigen_4area} plots the dominant poles of the four-converter test system corresponding to the above control settings (Cases~$1 \sim 3$). Compared to Case~1 (which has very low damping ratio), the damping ratio is increased to about 0.05 when applying grid-forming control in Converters~1 and 2, and is increased to about 0.4 when applying grid-forming control in Converters~3 and 4. The above eigenvalue results are fully aligned with the time-domain responses in Fig.~\ref{Fig_sim}, which again verifies the validity of the analysis in this paper and the effectiveness of the proposed algorithm for optimally placing grid-forming converters.

\begin{figure}[!t]
	\centering
	\includegraphics[width=2.6in]{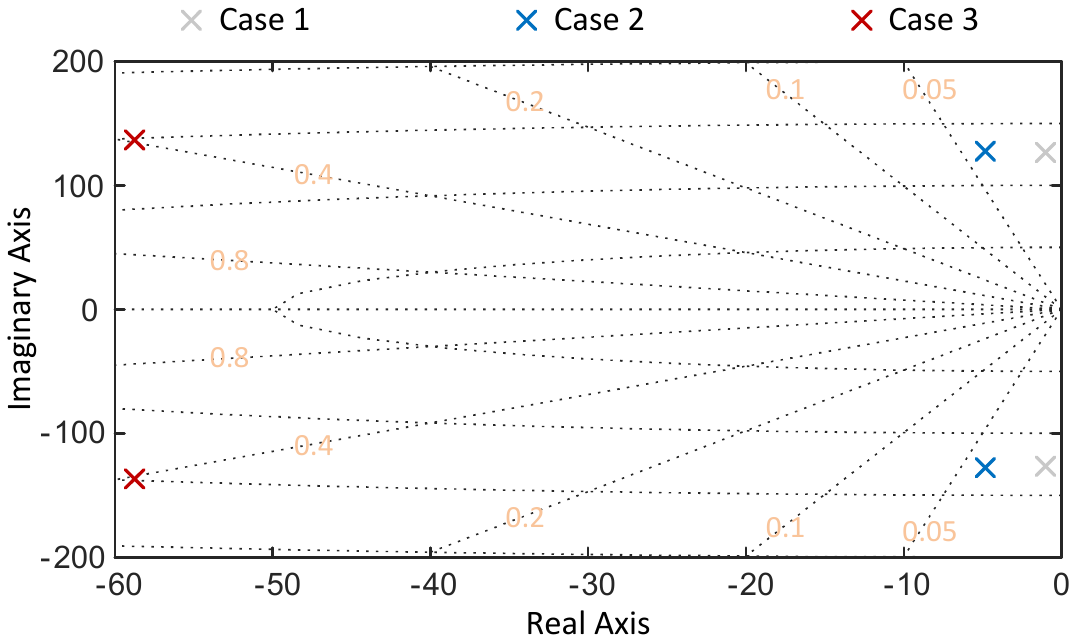}
	\vspace{-2mm}
	%\DeclareGraphicsExtensions.
	\caption{Dominant eigenvalues of the of the four-converter test system when placing grid-forming controls at different spots.}
	\vspace{0mm}
	\label{Fig_Eigen_4area}
\end{figure}

\subsection{Case studies on a nine-converter test system}

In the following, we consider a nine-converter system in order to further test the effectiveness of the proposed algorithm. As shown in Fig.~\ref{Fig_39bus}, The nine converters are interconnected via a 39-bus network, where Bus~39 is connected to an infinite bus. Note that the topology of the this network is actually the same as the standard England IEEE 39-bus system, and the network parameters are the same as those in \cite{huang2019impacts}. The converter parameters are basically the same as those listed in Appendix~\ref{Appendix: System parameters} (with the PLL parameters changed to $\{147, 10773\}$ in order to provide a critically stable base case).

\begin{figure}[!t]
	\centering
	\includegraphics[width=3in]{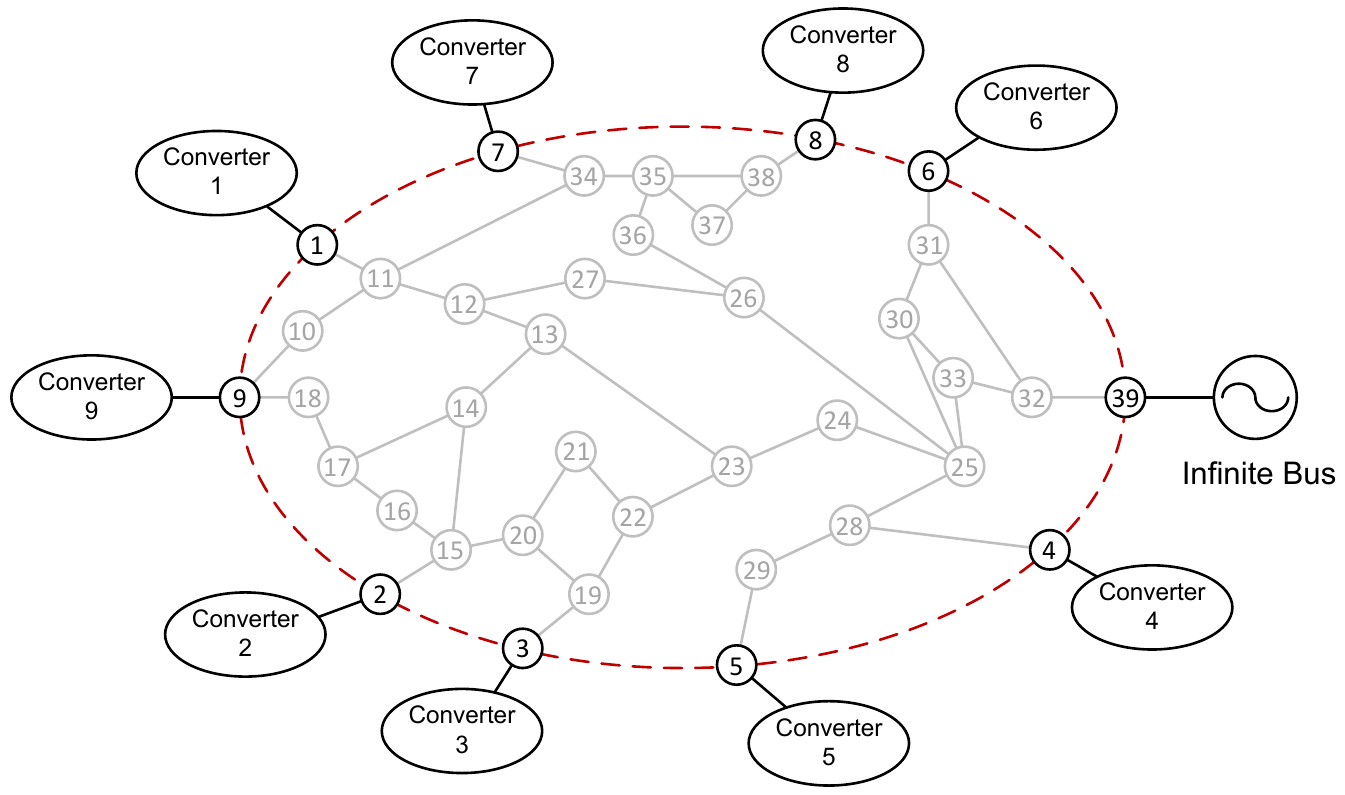}
	\vspace{-2mm}
	%\DeclareGraphicsExtensions.
	\caption{A nine-converter test system.}
	\vspace{0mm}
	\label{Fig_39bus}
\end{figure}

The smallest eigenvalue of the Kron-reduced Laplacian matrix of the network (with the network parameters given in \cite{huang2019impacts} and the capacities of the converters are identical) is $\lambda_1 = 3.31$. Then, we use the algorithms presented in Section~IV to choose two converters to be operated as grid-forming converters. To find the globally optimal solution, we solve the optimization problem in \eqref{eq:optimal_placement} by enumerating all the possible results and pick the best solution, and the obtained optimal set is $\mathcal{I} = \{1,4\}$. The smallest eigenvalue becomes $\lambda_{\min}[\mathcal{R}_{\mathcal{V} \backslash \mathcal{I}}(\mathcal{L})] = 14.43$ with the optimal set $\mathcal{I} = \{1,4\}$, which indicates that grid strength is dramatically increased by applying grid-forming controls in Converter~1 and Converter~4.

For comparison, we also use Algorithm~1 (wherein the solution to \eqref{eq:optimal_placement1} is obtained by checking the participation factors) to calculate the suboptimal solution to \eqref{eq:optimal_placement}, and the obtained locations are $\alpha_1 = 9$ and $\alpha_2 = 8$. The smallest eigenvalue becomes $\lambda_{\min}[\mathcal{R}_{\mathcal{V} \backslash \{\alpha_1,\alpha_2\}}(\mathcal{L})] = 12.03$ with this suboptimal set $\{\alpha_1,\alpha_2\}=\{8,9\}$. It can be seen that although this solution is a suboptimal one, it also effectively increases the grid strength, which validates the effectiveness of the proposed algorithm.

Fig.~\ref{Fig_39bus_sims} shows the active power responses of the nine converters under different control settings (an overload event occurs at 0.2s and is cleared after 0.02s). In Fig.~\ref{Fig_39bus_sims}~(a), all the converters apply PLL-based control and the system has very low damping ratio (low PLL-induced small-signal stability margin), which is caused by the interaction between PLL-base converters and weak grid conditions \cite{dong2018small, huang2019impacts}. By comparison, it can be seen from Fig.~\ref{Fig_39bus_sims}~(b) that when applying grid-forming controls in Converters 1 and 4 (which is the optimal solution as discussed above), the multi-converter system has high damping ratio and superior performance. In Fig.~\ref{Fig_39bus_sims}~(c), grid-forming controls are applied in Converters 8 and 9 (which is the suboptimal solution to \eqref{eq:optimal_placement} obtained by Algorithm~1 as discussed above), and it can be seen that the system also has superior damping performance, which again verifies the effectiveness of Algorithm~1. The above simulation results are again consistent with the analysis in this paper which shows the placements of grid-forming converters can effectively improve the (PLL-induced) small-signal stability of multi-converter systems.

\begin{figure}[!t]
	\centering
	\includegraphics[width=3.4in]{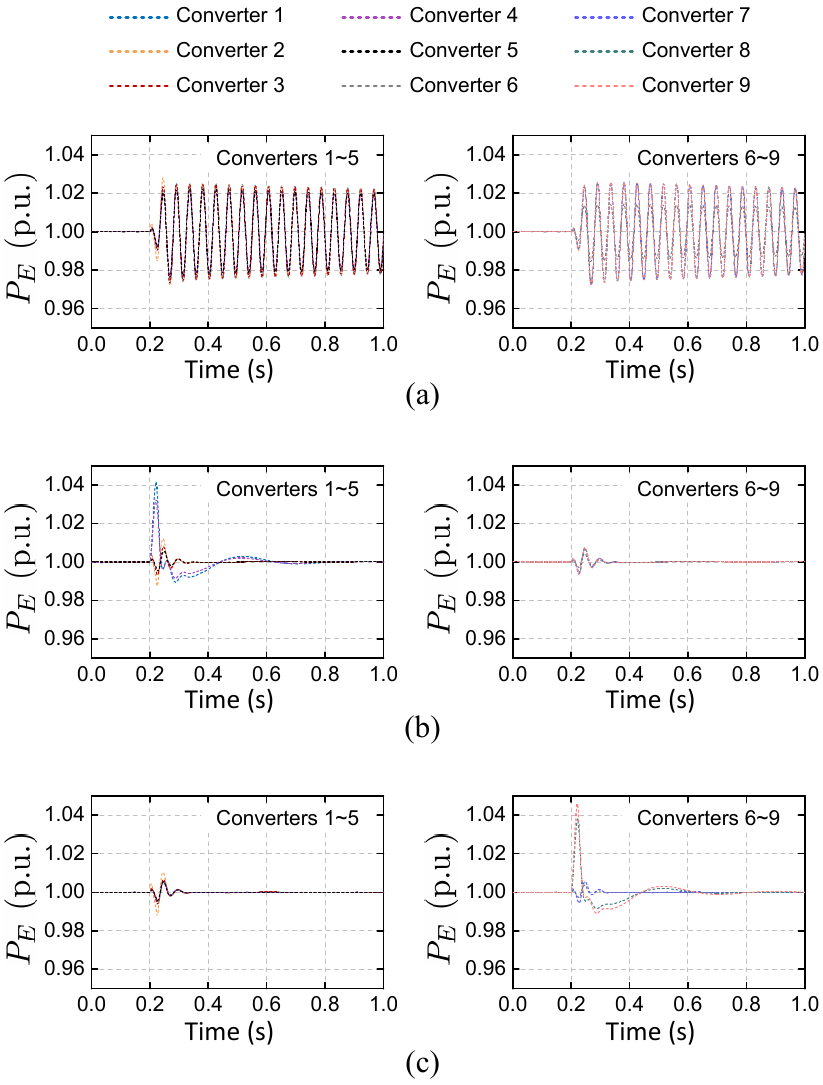}
	\vspace{-2mm}
	%\DeclareGraphicsExtensions.
	\caption{Time-domain responses of the nine converters when placing grid-forming controls at different spots (an overload event occurs at 0.2s and is cleared after 0.02s). (a) All the converters use PLL-based control. (b) Grid-forming control is applied in Converters 1 and 4. (c) Grid-forming control is applied in Converters 8 and 9.}
	\vspace{0mm}
	\label{Fig_39bus_sims}
\end{figure}

\section{Conclusions and Discussions}

This paper investigated the impacts of grid-forming converters on the small-signal stability of power systems that are integrated with PLL-based converters. By deriving the small-signal model of multi-converter systems, we explicitly revealed that the placement of grid-forming converters is equivalent to enhancing the power grid strength characterized by the generalized short-circuit ratio (i.e., the smallest eigenvalue of the weighted and Kron-reduced Laplacian matrix), thereby improving the (PLL-induced) small-signal stability. Our analyses in this paper focus on the characterization of power grid strength and thus provide a convenient way to study how the locations of grid-forming converters influence the system stability. On this basis, we investigated the optimal placement of grid-forming converters for improving the system stability. After explicitly formulating the optimization problem, we elaborated on how to solve it in a computationally efficient fashion such that it can be used in large-scale systems. Simulations based on a high-fidelity two-area test system and a high-fidelity 39-bus test system verify that placing grid-forming converters in proper locations can significantly improve the small-signal stability of PLL-integrated power systems.

Our work in this paper is based on characterizing the (PLL-induced) small-signal stability of multi-converter systems from the perspective of power grid strength, which provides a convenient way to rigorously show that the placements of grid-forming converters enhance the overall system stability (as discussed in Proposition~\ref{prop:Grid_Forming}). Moreover, by focusing on the power grid strength, the optimal locations to place grid-forming converters can be conveniently determined by solving an optimization problem which is related only to the network parameters and the capacities of the converters. In this way, we avoid directly optimizing the damping ratio of the dominant poles of the multi-converter system, which is intractable as it would require to derive the state matrix of the whole system.

Based on our theoretical analysis in this paper, possible future works can include the optimal placements of grid-forming converters considering both small-signal stability and frequency stability, and also stability-constrained network planning of multi-converter systems.

%\newpage

\bibliographystyle{IEEEtran}
\bibliography{ref}

%\renewcommand{\thetable}{\thesection.\arabic{table}}

%\setcounter{table}{0}

%See Table \ref{table:sys_parameters}.
\renewcommand\arraystretch{1.25}
\begin{table}%[h!]%h!
	\scriptsize
	\centering
	\vspace{0mm}
	\caption{Parameters of the two-area test system}
	\begin{tabular}{|llll|}
		\hline
        \multicolumn{4}{|c|}{Base Values for Per-unit Calculation}				\\
        \hline
        $f_{\rm base} = 50{\rm Hz}$		&	$\omega_{\rm base} = 2\pi f_{\rm base}$ &	$U_{\rm base} = 690{\rm V}$	&	$S_{\rm base} = 1.5{\rm MVA}$		\\
        \hline
		\multicolumn{4}{|c|}{Power Network Parameters (per-unit values)}				\\	
		\hline
        $B_{15} = {0.10}$ & $B_{25} = {0.05}$ & $B_{39} = {0.2}$ & $B_{49} = {0.15}$       \\
        $B_{56} = {0.015}$ & $B_{67} = {0.015}$ & $B_{78} = {0.035}$ & $B_{7,10} = {0.018}$       \\
		$B_{89} = {0.02}$ & $\tau = {0.1}$ & $C_{1} = {0.05}$ & $C_{2} = {0.05}$       \\
		\multicolumn{2}{|l}{Load 1: {0.5}}	& \multicolumn{2}{l|}{Load 2: {0.5}}	\\
        \hline
		\multicolumn{4}{|c|}{{\em LCL} Parameters of the Converter (per-unit values)}			\\
		\hline
		$L_F = 0.05$		&	$C_F = 0.05$		&	$L_g = 0.06$			&	{}		\\
		\hline
		\multicolumn{4}{|c|}{Parameters of the PLL-Based Control (per-unit values)}		\\
		\hline
		\multicolumn{4}{|l|}{PI parameters of the current control Loop: 0.3, 10}		\\
        \multicolumn{4}{|l|}{PI parameters of the active power control Loop: 0.5, 40}		\\
        \multicolumn{4}{|l|}{PI parameters of the reactive power control Loop: 0.5, 40}		\\
        \multicolumn{4}{|l|}{PI parameters of the PLL: 104, 5390}		\\
        $P^{\rm ref} = 1.0$		&	$Q^{\rm ref} = 0$		&	$T_{\rm VF} = 0.02$			&	$K_{\rm VF} = 1.0$		\\
		\hline
		\multicolumn{4}{|c|}{Parameters of the Grid-Forming Control (per-unit values)}									\\
		\hline
		\multicolumn{4}{|l|}{PI parameters of the current control Loop: 0.3, 10}		\\
        \multicolumn{4}{|l|}{PI parameters of the voltage control Loop: 4, 30}		\\
        $T_{\rm VF} = 0.02$			&	$K_{\rm VF} = 1.0$		&	$J = 2$			&	$D = 50$		\\
        $k_F = 0.1$			&	$P_0 = 1.0$		&	{}			&	{}		\\
		\hline
	\end{tabular}	
    \vspace{15mm}
	\label{table:sys_parameters}
\end{table}

\appendices
\vspace{0mm}
\section{System Parameters}
\label{Appendix: System parameters}

See Table~\ref{table:sys_parameters}.

\end{document}